\documentclass[12pt]{article}

\usepackage{geometry}
\usepackage{calc}

\usepackage{amsthm,amsmath,amssymb}
\usepackage{graphicx}
\usepackage{enumerate}
\usepackage[dvipsnames]{xcolor}
\definecolor{nblue}{RGB}{19, 39, 150}
\definecolor{dark-gray}{gray}{0.35}
\usepackage[colorlinks=true,citecolor=nblue,linkcolor=nblue,urlcolor=blue]{hyperref}
\usepackage{mathtools}
\usepackage{bbm}  

\allowdisplaybreaks

\newcommand{\arxiv}[1]{\href{http://arxiv.org/abs/#1}{\texttt{arXiv:#1}}}

\theoremstyle{plain}
\newtheorem{theorem}{Theorem}
\newtheorem{lemma}[theorem]{Lemma}
\newtheorem{corollary}[theorem]{Corollary}
\newtheorem{proposition}[theorem]{Proposition}

\theoremstyle{definition}

\theoremstyle{remark}
\newtheorem{remark}[theorem]{Remark}

\usepackage{tikz}
\usetikzlibrary{arrows,automata,positioning,math}
\usepackage[latin1]{inputenc}
\usepackage{pgfplots}

\tikzmath{\dist = 0.3; \labelsize = 0.8; \nodesize = 0.4;}

\usepackage{comment}
\usepackage{relsize}

\newcommand{\N}{\mathbb{N}}

\newcommand{\set}[1]{\{#1\}}

\DeclarePairedDelimiter{\abs}{\lvert}{\rvert}
\DeclarePairedDelimiter{\paren}{(}{)}

\DeclarePairedDelimiterX{\setsuchthat}[2]{\lbrace}{\rbrace}{#1\
\delimsize\vert\ #2} 

\usepackage{ifthen}
% ...

\numberwithin{theorem}{section}

\def\nfrac#1#2{{\textstyle\frac{#1}{#2}}}
\def\dfrac#1#2{\lower0.15ex\hbox{\large$\frac{#1}{#2}$}}

\def\dvec{\boldsymbol{d}}
\def\dmax{d_{\max}}

\def\kvec{\boldsymbol{k}}
\def\kmax{k_{\max}}

\numberwithin{equation}{section}

\def\({\bigl(}
\def\){\bigr)}

\title{Sampling hypergraphs with given degrees}

\author{Martin Dyer\footnote{School of Computing, University of Leeds, Leeds LS2 9JT, UK. \texttt{m.e.dyer@leeds.ac.uk}}
\and Catherine Greenhill\footnote{School of Mathematics and Statistics, UNSW Sydney, NSW 2052, Australia. \texttt{c.greenhill@unsw.edu.au}, \texttt{james.ross@unsw.edu.au}}
\and Pieter Kleer\footnote{Tilburg University, Tilburg, The Netherlands. \texttt{p.s.kleer@tilburguniversity.edu}}\\
\and James Ross\textsuperscript{$\dagger$} 
\and Leen Stougie\footnote{CWI, INRIA-Erable, Vrije Universiteit,
Amsterdam, The Netherlands. \texttt{Leen.Stougie@cwi.nl}} }

\date{13 July 2021}

\begin{document}

\maketitle

\begin{abstract}
There is a well-known connection between hypergraphs and bipartite graphs,
obtained by treating the incidence matrix of the hypergraph as the
biadjacency matrix of a bipartite graph.  We use this connection to
describe and analyse a rejection sampling algorithm for sampling
simple uniform hypergraphs with a given degree sequence.
Our algorithm uses, as a black box, an algorithm $\mathcal{A}$ for sampling
bipartite graphs with given degrees, uniformly or nearly uniformly,
in (expected) polynomial time.  The expected
runtime of the hypergraph sampling algorithm depends on the (expected) runtime
of the bipartite graph sampling algorithm $\mathcal{A}$, and the probability
that a uniformly random bipartite graph with given degrees
corresponds to a simple hypergraph.
We give some conditions on the hypergraph degree sequence which guarantee that this probability is bounded below by a positive constant.

\medskip

\end{abstract}

\section{Introduction}
Hypergraphs are combinatorial objects which can be used to abstractly
represent general dependence structures, with applications in
many areas including machine learning \cite{ZHS2007} and
bioinformatics \cite{PK2013}.  We consider the problem
of efficiently sampling simple, $k$-uniform hypergraphs with a given degree
sequence, either uniformly or approximately uniformly.

More precisely, a hypergraph $H = (V,E)$ consists of a finite set $V = V(H)=
\{v_1,\dots,v_n\}$ of $n$ nodes, and a multiset $E=E(H) = \{e_1, e_2, \ldots, e_m\}$ of $m$ 
edges, where each edge is a nonempty multisubset of $V$.
We say that $H$ is \emph{simple} if there are no repeated edges in $E$
and no edge of $E$ contains a repeated node (so $E$ is a set of subsets of nodes).
For any node $v_i \in V$, we define the degree of $v_i$ by
  \[
    d_i = \deg_H(v_i)
      = |\{ e \in E(H) : v_i \in e \}|,
  \]
and write $\boldsymbol{d} = (d_1, \ldots, d_n)$ for the degree
sequence of $H$.
For a positive integer $k$, we say that $H$ is $k$-\emph{uniform} if
every edge contains exactly $k$ nodes, counting multiplicities when $H$ is not simple.
We then write
$\mathcal{H}_k(\boldsymbol{d})$ for the set of $k$-uniform simple
hypergraphs with degree sequence $\boldsymbol{d}$. If there is a
positive integer $d$ such that $d_i = d$ for all $i \in [n]$, then
we write $\mathcal{H}_k(n,d)$ for the set of all $k$-uniform
$d$-regular simple hypergraphs on $n$ nodes, instead of
$\mathcal{H}_k(\boldsymbol{d})$.

Recently, Deza, Levin, Meesum \& Onn~\cite{DLMO} proved that the construction
problem for simple 3-uniform hypergraphs is NP-hard.  That is, given
$\dvec$ it is NP-hard to decide whether there exists a 3-uniform simple
hypergraph with degree sequence $\dvec$.  This implies that it is
not possible to approximate $|\mathcal{H}_k(\dvec)|$
efficiently in general, since approximate counting can distinguish 0
from a positive number.
Moreover, hardness of construction also directly implies that (approximate) uniform sampling is
a difficult problem in general.

Arafat et al. \cite{ABDB2020} recently gave an algorithm to construct a non-simple hypergraph with given degrees and edge sizes. Chodrow \cite{C2019} considered Markov Chain Monte Carlo approaches for generating such hypergraphs.  We emphasize that throughout this work, we only consider simple hypergraphs.
To the best of our knowledge, the only rigorously-analysed algorithm
for this problem in the literature is the configuration model, see
Section~\ref{s:sampling-hypergraphs}.

Our approach is based on the well-known connection between hypergraphs
and bipartite graphs.  We first fix some notation for
bipartite graphs and then explain this relation. A bipartite graph $B = (X \cup Y, A)$ consists of a bipartition $X =
\{x_1,\dots,x_n\}$ and $Y = \{y_1,\dots,y_m\}$ of nodes, and an edge
set $A \subseteq X\times Y = \{\{x,y\} : x \in X, y \in Y\}$.
For a pair of nonnegative integer sequences $\boldsymbol{d} = (d_1,
\ldots, d_n)$ and $\boldsymbol{k} = (k_1, \ldots, k_m)$, let
$\mathcal{B}(\boldsymbol{d},\boldsymbol{k})$ be the set of all
simple bipartite graphs $B$ such that
  \[
    \deg_B(x_i) = d_i \textrm{ for all } x_i \in X,
    \quad \textrm{and} \quad
    \deg_B(y_j) = k_j \textrm{ for all } y_j \in Y.
  \]
We say that $(\boldsymbol{d}, \boldsymbol{k})$ is the (bipartite)
degree sequence of $B$.
Note that if $\sum_{i=1}^n d_i \neq \sum_{j=1}^m k_j$, then $\mathcal{B}(\boldsymbol{d},\boldsymbol{k})=\emptyset$. 
If there is a fixed integer $k$ such that $k_j = k$ for all $j \in
[m]$, then we write $\mathcal{B}(\boldsymbol{d}, k)$ instead of
$\mathcal{B}(\boldsymbol{d},\boldsymbol{k})$, and we call such
bipartite graphs \emph{half-regular}. If in addition
$\dvec=(d,d,\ldots, d)$ is regular then we write
$\mathcal{B}(n,d,k)$. For any node $v \in X \cup Y$, let
$\mathcal{N}_B(v) = \{w \in X \cup Y : \{v,w\} \in A\}$ be the
neighbourhood of node $v$ in $B$.

Every hypergraph $H = (V, E)$ can be represented as a
bipartite graph $B_H$, as follows.
Fix a labelling of the edges of $H$, say $E = \{e_1,\ldots, e_m\}$,
then let
\[ X = V,\quad Y = E \quad \text{ and }
  \quad A = \Big\{ \{v_i, e_j\}\in X\times Y \mid v_i \in e_j \Big\}.
\]
If $H\in \mathcal{H}_k(\boldsymbol{d})$ then $B(H)\in\mathcal{B}(\dvec,k)$.
Conversely, every bipartite graph $B = (X \cup Y,A)$
corresponds to a hypergraph $H_B = (V,E)$, where $V = X$
and
\[ E = \{\mathcal{N}_B(y) \mid y\in Y\}.\]
Furthermore, $H_B$ is simple if and only if every node in $Y$ has a distinct
set of neighbours in $B$; that is, if $\mathcal{N}_B(y_i)= \mathcal{N}_B(y_j)$ implies $i=j$.
If $H_B$ is simple then we say that the bipartite graph $B$ is \textit{H-simple}.

Write $\mathcal{B}^*(\dvec,k)$ for the set of all H-simple
half-regular bipartite graphs, and define
 $\varphi : \mathcal{B}^*(\boldsymbol{d}, k) \rightarrow \mathcal{H}_k(\boldsymbol{d})$ as the canonical mapping that maps $B$ to
 $H_B$, as described above.
 We can use rejection sampling to turn any sampling algorithm
for $\mathcal{B}(\dvec,k)$ into a sampling algorithm for
$\mathcal{H}_k(\dvec)$, as follows:

\begin{center}
\hspace{\dimexpr-\fboxrule-\fboxsep\relax}\fbox{
\begin{minipage}{0.8\textwidth}
\begin{tabbing}
    \textsc{HypergraphSampling}$(\dvec,k,\mathcal{A})$\\
    XX\= XX\= XX\=  \kill
    \textsc{Input}: \>\>\, Parameters $(\dvec,k)$;\\
            \>\> algorithm $\mathcal{A}$ for sampling from
                $\mathcal{B}(\dvec,k)$.\\
    \emph{Begin}\\
    \> \emph{repeat}\\
    \>\>  sample $B$ from $\mathcal{B}(\dvec,k)$ using $\mathcal{A}$\\
    \>\> \emph{until}\,  $B$ is H-simple\\
    \> output $\varphi(B)$\\
    \emph{end.}
\end{tabbing}
\end{minipage}}
\end{center}

\noindent Note that for all $H\in\mathcal{H}_k(\boldsymbol{d})$ we have
$|\varphi^{-1}(H)|=m!$, as there are $m!$ distinct ways to label the
edges of $H$ when $H$ is simple. Hence, if $\mathcal{A}$ samples
\emph{uniformly} from $\mathcal{B}(\dvec,k)$ then the output of
\textsc{HypergraphSampling} is  uniform over $\mathcal{H}_k(\dvec)$.

The expected number of times \textsc{HypergraphSampling} draws a
sample from $\mathcal{B}(\dvec,k)$, using algorithm $\mathcal{A}$,
depends on the proportion of bipartite graphs in $\mathcal{B}(\dvec,k)$
which are H-simple: that is, the ratio
$|\mathcal{B}^*(\dvec,k)|/|\mathcal{B}(\dvec,k)|$.
The goal of our work is to identify pairs $(\dvec,k)$ for which
$|\mathcal{B}(\dvec,k)|/|\mathcal{B}^*(\dvec,k)|$ is bounded above
by a polynomial in $n$.
For such pairs, if the output of $\mathcal{A}$ is close to uniform then
this implies that the expected number of times
we run $\mathcal{A}$ before an H-simple element of
$\mathcal{B}(\dvec,k)$ is found is at most polynomial. We also identify pairs 
$(\dvec, k)$ for which the ratio is bounded above by a constant, in which case only a 
constant number of calls to $\mathcal{A}$ are expected. 
This is made more specific in the next subsection.

\subsection{Our contributions}\label{subsec:our_contributions}

The total variation distance between two probability distributions $\sigma$, $\pi$
on a set $\Omega$ is given by
\begin{equation}
    d_{TV}(\sigma,\pi)
    = \frac{1}{2}\sum_{x\in\Omega} |\sigma(x) - \pi(x)|
    = \max_{ S\subseteq \Omega}\, |\sigma(S) - \pi(S)|.
\label{dTV}
\end{equation}
Suppose that the distribution of the output of algorithm
$\mathcal{A}$ on $\mathcal{B}(\dvec,k)$ is $\sigma_\mathcal{B}$, and
let $\sigma_{\mathcal{H}}$ be the output of the algorithm
\textsc{HypergraphSampling}.  Then $\sigma_{\mathcal{H}}$ is a
distribution on $\mathcal{H}_k(\dvec)$ which is obtained by setting
$H=\varphi(B)$ where $B$ has distribution $\sigma_{\mathcal{B}}$,
conditioned on $B\in\mathcal{B}^*(\dvec,k)$.
To make it clear which distribution we are using, we write
$\mathbb{P}_{\sigma}$ for the probability mass function of the distribution $\sigma$.
Let $\pi_{\mathcal{B}}$ be the uniform distribution on
$\mathcal{B}(\boldsymbol{d},k)$, and let $\pi_{\mathcal{H}}$ be the
uniform distribution on $\mathcal{H}_k(\boldsymbol{d})$.

A \emph{fully-polynomial almost uniform sampler (FPAUS)} for sampling from
a set $\Omega$ is an algorithm that, with probability at least $\nfrac{3}{4}$,
outputs an element of $\Omega$ in time polynomial in $\log|\Omega|$ and $\log(1/\varepsilon)$, such that the output distribution is \emph{$\varepsilon$-close} to the
uniform distribution $\pi$ on $\Omega$ in \emph{total variation distance}:
that is, $d_{TV}(\sigma,\pi) \leq \varepsilon$.
If $\Omega = \mathcal{B}(\dvec,k)$ or $\Omega = \mathcal{H}_k(\dvec)$ then
$\log |\Omega| = O(M\log M)$ for $M = \sum_{i=1}^n d_i$:  this follows from~\cite[Theorem 1.3]{GMW},
restated below as Theorem~\ref{thm:zero_one_matrix_count_precise}.
So an FPAUS for $\mathcal{H}_k(\dvec)$ or $\mathcal{B}(\dvec,k)$ must
have running time bounded above by a polynomial in $M$
and $\log(1/\varepsilon)$.

Denote by $\rho(\dvec,k)$ the runtime of an algorithm for testing H-simplicity, when
run on inputs from $\mathcal{B}(\dvec,k)$.  
Given $B\in\mathcal{B}(\dvec,k)$, we can test whether $\varphi(B)$ is
H-simple 
by creating the edges of $\varphi(B)$ and sorting them lexicographically,
requiring $O(M\log M)$ time and $O(M)$ space. (Better implementations may
be possible.)   
Hence we can assume that $\rho(\dvec,k) = O(M\log M)$, and that the
H-simplicity test also creates the hypergraph $\varphi(B)$.
Our first result summarises the properties of the algorithm
\textsc{HypergraphSampling}$\ (\dvec,k,\mathcal{A})$ in terms of
the output distribution and runtime of~$\mathcal{A}$, and the runtime
$\rho(\dvec, k)$ of the H-simplicity test.
The proof of Theorem~\ref{thm:algorithm}, which is fairly standard, is presented in Section~\ref{sec:rejection_sampling}.

\begin{theorem}
\label{thm:algorithm}
Suppose that $n$ is a positive integer, $\boldsymbol{d}=(d_1,\ldots, d_n)$ is a
sequence of positive integers, and $k$ is a positive integer such
that $k \geq 3$ and $\mathcal{B}^*(\boldsymbol{d},k)$ is non-empty.
\begin{itemize}
\item[\emph{(i)}]
The output distribution $\sigma_{\mathcal{H}}$ of
\textsc{HyperGraphSampling} satisfies
  \[
    d_{TV}\paren*{ \sigma_{\mathcal{H}}, \pi_{\mathcal{H}} }
      \leq \frac{3}{2}
        \cdot \frac{
  d_{TV}\left( \sigma_{\mathcal{B}},\pi_{\mathcal{B}}\right)
        }{
          \mathbb{P}_{\pi_{\mathcal{B}}} \paren*{\mathcal{B}^*(\boldsymbol{d},k)
          }
        }.
  \]
\item[\emph{(ii)}]
The expected runtime of
\textsc{HypergraphSampling}$(\dvec,k,\mathcal{A})$ is 
  \[ 
    q(\dvec,k)\, \Big( \tau(\dvec,k) + \rho(\dvec, k) \Big) 
  \] 
where $q(\dvec,k)^{-1} = \mathbb{P}_{\sigma_{\mathcal{B}}} \paren*{
\mathcal{B}^*(\dvec,k)}$ is the probability that a sampled bipartite graph is 
H-simple and 
$\tau(\dvec,k)$ is the (expected) runtime of algorithm $\mathcal{A}$ on 
$\mathcal{B}(\dvec,k)$. 
Furthermore, the probability that
\textsc{HyperGraphSampling} needs more than $t\, q(\dvec,k)$
iterations of $\mathcal{A}$ before finding an element of
$\mathcal{B}^*(\dvec,k)$ is at most $\exp(-t)$\, for any $t>0$.
\item[\emph{(iii)}]
Suppose that $d_{TV}\left( \sigma_{\mathcal{B}},\pi_{\mathcal{B}}\right) \leq \varepsilon$
and
         $\mathbb{P}_{\pi_{\mathcal{B}}} \paren*{\mathcal{B}^*(\boldsymbol{d},k)}\geq 1-c_0$\,
for some $\varepsilon \in (0,1)$ and $c_0\in (0,1-\varepsilon)$. Then
  \[ 
    d_{TV}(\sigma_{\mathcal{H}},\pi_{\mathcal{H}}) \leq \frac{3\varepsilon}{2(1-c_0)}
  \]
and the expected runtime of
\textsc{HypergraphSampling}$(\dvec,k,\mathcal{A})$ is at most
  \[ 
    (1-c_0-\varepsilon)^{-1}\, \Big(\tau(\dvec,k) + \rho(\dvec, k) \Big). 
  \]
\item[\emph{(iv)}]
If $\mathcal{A}$ is an FPAUS for $\mathcal{B}(\dvec,k)$ and the assumptions of (iii) 
hold, then we can transform \textsc{HypergraphSampling}$(\dvec,k,\mathcal{A})$ into an 
FPAUS for $\mathcal{H}_k(\dvec)$ by terminating after
\mbox{$\lceil 2(1-c_0-\varepsilon)^{-1}\rceil$} calls to $\mathcal{A}$
and reporting \emph{\texttt{FAIL}}.
\end{itemize}
\end{theorem}

\bigskip

\begin{remark}
For Theorem~\ref{thm:algorithm}(iv) to provide an FPAUS with an
\emph{explicit} upper bound on the runtime, explicit bounds are
needed on $\varepsilon$, $c_0$ and $\rho(\dvec,k)$. However, for
Theorem~\ref{thm:algorithm}(iii) it is enough to know that
sufficiently small values of $c_0$, $\varepsilon$ exist. 
\end{remark}

\begin{remark}  
We require $O(M)$ space to store a bipartite graph $B\in\mathcal{B}(\dvec,k)$,
or the corresponding hypergraph $\varphi(B)$. It follows that 
the space complexity of \textsc{HypergraphSampling} is bounded above by the sum 
of $O(M)$ and the 
space complexity of the chosen algorithm $\mathcal{A}$. 
\end{remark}


\bigskip

We see from Theorem~\ref{thm:algorithm} that
$d_{TV}(\sigma_{\mathcal{B}},\pi_{\mathcal{B}})$ and
$\mathbb{P}_{\pi_{\mathcal{B}}} \paren*{\mathcal{B}^*(\dvec,k)}$
are the two crucial quantities which control both the expected runtime of
\textsc{HypergraphSampling}$(\dvec,k,\mathcal{A})$, and how far the output varies
from uniform.
The first of these quantities is determined by the choice of algorithm $\mathcal{A}$.
In our next two theorems, we provide a lower bound on
$\mathbb{P}_{\pi_{\mathcal{B}}} \paren*{\mathcal{B}^*(\dvec,k)}$
when $\dvec = (d,d,\ldots, d)$ is regular, and give
an asymptotic lower bound on this probability which holds when $\dvec$
is irregular but sparse.

\begin{theorem}
  \label{thm:main_regular}
  Let $n \in \N$, $d \in \N$ and $k \geq 3$ so that $m = nd/k \in \N$ and
$\binom{m}{2} < \binom{n}{k}$.  Then
  \[
    \mathbb{P}_{\pi_{\mathcal{B}}} \paren*{\mathcal{B}^*(n, d, k)}
    \geq 1 - \binom{m}{2}\binom{n}{k}^{-1}.
  \]
Hence Theorem~\emph{\ref{thm:algorithm}(iii)} applies when $\binom{m}{2} \leq c_0\binom{n}{k}$
for some $c_0\in (0,1-\varepsilon)$, where
$d_{TV}\left( \sigma_{\mathcal{B}},\pi_{\mathcal{B}}\right) \leq \varepsilon$.
\end{theorem}

\begin{remark} 
When $k \geq 3$ is a fixed constant, the lower bound in
Theorem~\ref{thm:main_regular} is $1-o(1)$ if $d = o(n^{k/2-1})$.
\end{remark}

\bigskip

\begin{theorem}
\label{thm:main_irregular}
Assume that for each positive
integer $n$ we have an integer $k=k(n)\geq 3$ and a sequence
$\boldsymbol{d} =\boldsymbol{d}(n) = (d_1, \ldots, d_n)$ of positive
integers such that $M = \sum_{i=1}^n d_i$ tends to infinity
with $n$.  Suppose that $k$ divides $M$ for infinitely many values of $n$.
Assume that
$k^2\dmax^2 = o(M)$ 
and let $\pi_{\mathcal{B}}$ be the uniform distribution on
$\mathcal{B}(\boldsymbol{d},k)$. Then
  \[
    \mathbb{P}_{\pi_{\mathcal{B}}} \paren*{
      \mathcal{B}^*(\boldsymbol{d}, k)
    }
    \geq 1 -
      \frac{n^k \, \dmax^k}{M^k}
      \cdot \binom{M/k}{2} \, \binom{n}{k}^{-1}
      \cdot (1 + o(1)).
  \]
Writing $m=M/k$ and $d = M/n$, we see that Theorem~\ref{thm:algorithm}(iii) applies when
\begin{equation}
\label{eq:c0}
 \left(\frac{\dmax}{d}\right)^k\, \binom{m}{2} \leq c_0 \binom{n}{k}
\end{equation}
for some $c_0\in (0,1-\varepsilon)$, where
$d_{TV}\left( \sigma_{\mathcal{B}},\pi_{\mathcal{B}}\right) \leq \varepsilon$.
\end{theorem}

\bigskip

In the hypergraph setting, $m$ is the number of edges and $d$ is the average degree
of any hypergraph in $\mathcal{H}_k(\dvec)$.

\begin{remark}
\label{rem:irreg}
When $k\geq 3$ is a fixed constant, the lower bound in
Theorem~\ref{thm:main_irregular} is $1-o(1)$ if $\dmax = o(M^{1-2/k})$.
Similarly, if $k\geq 3$ is a fixed constant and $\dmax=O(d)$
then the lower bound in Theorem~\ref{thm:main_irregular} is $1-o(1)$
whenever $d^2 = o(n^{k-2})$, as in the regular case.
\end{remark}

\bigskip

There are several algorithms $\mathcal{A}$ in the literature
for efficiently sampling
from $\mathcal{B}(\dvec,k)$, either uniformly or almost
uniformly, under various conditions on $\dvec$ and $k$.
These will be reviewed in Section~\ref{s:bipartite-samplers}, together
with the properties of the resulting algorithm
\textsc{HypergraphSampling}($\dvec, k,\mathcal{A})$.
In Section~\ref{s:sampling-hypergraphs} we discuss the configuration
model for hypergraphs, which can be used as an expected polynomial-time
sampling algorithm  when $k\dmax = O(\log n)$.

Then in Section \ref{sec:rejection_sampling} we provide a general framework which we use to analyse the algorithm \textsc{HypergraphSampling}. In Section \ref{sec:H-simple} we fill in the details for the regular regime (Theorem \ref{thm:main_regular}) and the irregular, sparse regime (Theorem \ref{thm:main_irregular}).

\section{Related work}\label{s:related}

\subsection{Various bipartite sampling algorithms, and implications}
\label{s:bipartite-samplers}

In this section we describe several algorithms for efficient
sampling from $\mathcal{B}(\dvec,k)$, uniformly or almost uniformly,
under various conditions on $\dvec$, $k$. We also apply
Theorem~\ref{thm:algorithm} to describe when
\textsc{HypergraphSampling}($\dvec, k,\mathcal{A})$ is an efficient
algorithm for sampling from $\mathcal{H}_k(\dvec)$ (uniformly or
near-uniformly), for each choice of $\mathcal{A}$.
We focus on the time complexity of these algorithms, as most authors
have not stated the space complexity of their algorithms.

\bigskip

\noindent The first two algorithms mentioned below perform \emph{exactly uniform} sampling from $\mathcal{B}(\dvec,k)$.

\begin{itemize}
\item[(I)] If $k\dmax \leq C\log n$ then the bipartite configuration model
gives rise to an algorithm for sampling (exactly) uniformly from
$\mathcal{B}(\dvec,k)$.  But the bipartite configuration model is
equivalent to the configuration model for hypergraphs, described
in Section~\ref{s:sampling-hypergraphs}, and so there
is no advantage to working in the bipartite graph setting
when $k\dmax\leq C\log n$.
(See Lemma~\ref{lem:hypergraph-config}).
\item[(II)] Next suppose that $(\dmax + k)^4 = O(M)$.
Building on the work of~\cite{GW,MW90},
Arman, Gao and Wormald~\cite[Theorem~4]{AGW} describe an algorithm
which samples uniformly from $\mathcal{B}(\dvec,k)$ with expected runtime $O(M)$
and space complexity $O(nM)$.
Note that $d_{TV}(\sigma_{\mathcal{B}},\pi_{\mathcal{B}})=0$ as
the output is exactly uniform.  Using this algorithm as $\mathcal{A}$
and applying Theorem~\ref{thm:algorithm},
we see that  \textsc{HypergraphSampling}($\dvec, k,\mathcal{A})$
performs exact sampling from $\mathcal{H}_k(\dvec)$
with expected runtime $O(M + \rho(\dvec,k))$ whenever (\ref{eq:c0}) holds for some
constant $c_0\in (0,1)$.  In particular, this holds whenever $(\dvec,k)$
are as described in Remark~\ref{rem:irreg}.
\end{itemize}

\noindent The next algorithm produces output which is \emph{asymptotically uniform},
meaning that the output distribution
is only $o(1)$ from uniform in total variation distance.

\begin{itemize}
\item[(III)] If $\dmax + k = O(M^{1/4-\tau})$ for some positive constant $\tau$
then the algorithm from~\cite{AGW} can be applied, as described in (II).
 An alternative is to use the sampling algorithm of
Bayati, Kim and Saberi~\cite[Theorem~1]{BKS}, which has expected runtime $O(\dmax M)$
(see the proof of~\cite[Theorem 1]{BKS}).  The output of this algorithm
satisfies $d_{TV}(\sigma_{\mathcal{B}},\pi_{\mathcal{B}}) = o(1)$,
where this vanishing term depends only on $n$ and cannot be made smaller
by increasing the runtime of the algorithm.  Hence we can take
$\varepsilon = o(1)$ in Theorem~\ref{thm:algorithm}, and conclude that
for this choice of $\mathcal{A}$, the algorithm
\textsc{HypergraphSampling}($\dvec, r,\mathcal{A})$ has expected
runtime $O(\dmax M + \rho(\dvec,k))$ whenever (\ref{eq:c0}) holds for some constant
$c_0\in (0,1)$, and the distribution of the output is within $o(1)$
of uniform:  that is, $d_{TV}(\sigma_{\mathcal{H}},\pi_{\mathcal{H}})=o(1)$.
\end{itemize}

\medskip

\begin{remark}
Although the Arman, Gao and Wormald algorithm applies for a slightly
wider range of values of $(\dvec,k)$, has a better expected runtime
bound and performs exactly uniform sampling, the Bayati, Kim and
Saberi algorithm has one advantage: it is much easier to implement.
Indeed, Bayati, Kim and Saberi~\cite[Theorem~3]{BKS} used sequential importance
sampling to give an algorithm which is close to an FPAUS,
except that the runtime is polymial in $1/\varepsilon$, while in
an FPAUS the dependence should be on $\log(1/\varepsilon)$.
However, this algorithm
is valid only when $\dmax = O(M^{1/4- \tau})$ for some $\tau >0$
and no longer has the advantage of simplicity, and so it is
surpassed by the fast, precisely uniform sampling algorithm
of Arman, Gao and Wormald~\cite{AGW}, described in (II) above.
(Other authors, such as Chen et al.~\cite{chen}, have
 used sequential importance sampling to sample
bipartite graphs with given degrees, but without rigorous analysis.)
\end{remark}

\medskip

Now we survey some algorithms which are FPAUSs for $\mathcal{B}(\dvec, k)$.
Each can be used as $\mathcal{A}$ to give an FPAUS for
$\mathcal{H}_k(\dvec)$, using Theorem~\ref{thm:algorithm}(iv),
so long as (\ref{eq:c0}) holds
(or $\binom{m}{2}\leq c_0\binom{n}{k}$ when $\dvec$ is regular)
for some $c_0\in (0,1-\varepsilon)$.
The runtime of each of these algorithms is the mixing time of their underlying Markov chain multiplied by the cost of performing a single step of the chain, the latter being an implementation-dependent cost.
In all cases, the polynomial bound on the mixing time is quite a high-degree
polynomial and is not believed to be tight.  We do not always state
the mixing time.

\begin{itemize}
\item[(IV)]
Jerrum, Sinclair and Vigoda~\cite{JSV} described and analysed a simulated
annealing algorithm which gives an FPAUS for sampling perfect
matchings from a given bipartite graph. As a corollary, they obtained
an FPAUS for sampling bipartite graphs for given degrees~\cite[Corollary~8.1]{JSV}.
Bez{\'a}kov{\'a}, Bhatnagar and Vigoda \cite{BBV} adapted the algorithm
from~\cite{JSV} to provide a simplified
FPAUS for $\mathcal{B}(\dvec,k)$, valid for any $(\dvec,k)$,
with running time
\[ O((nm)^2 M^3 \Delta\, \log^4(nm/\varepsilon)), \]
where $n$ and $m$ are the number of nodes in each part of the bipartition,
and $\Delta = \max\{\dmax,k\}$.

\item[(V)]
Another well-studied Markov chain for sampling (bipartite) graphs
with given degrees is the switch Markov chain. It is the
simplest Markov chain which walks on the set of all (bipartite)
graphs with a given degree sequence, as it deletes and replaces only two
edges at a time. 
The chain has been used in many contexts, including
contingency tables~\cite{DG}, and was first applied to
bipartite graphs by Kannan, Tetali and Vempala \cite{KTV}.
The mixing time of the switch chain has been shown
to be polynomial for various bipartite and
general degree sequences, see for example \cite{AK,CDG,hungarians,GS,MES,TY2020preprint}.
If a Markov chain leads to an FPAUS then we say that the Markov chain
is rapidly mixing.
In particular
\begin{itemize}
\item Tikhomirov and Youssef \cite{TY2020preprint} considered the switch Markov chain on regular bipartite graphs (in which $d = k$), and proved a sharp Poincar\'e inequality and an optimal upper bound on the log-Sobolev constant for the switch chain. With these results, they demonstrated that for some constant $c > 0$, when $3 \leq d \leq n^c$ the mixing time of the switch Markov chain on $d$-regular bipartite graphs is at most
  \[
    O \left( dn(dn \log dn + \log(2\varepsilon^{-1}) \right).
  \]
This is a significant improvement on the previously known bounds in other settings.
\item Miklos, Erd\H{o}s and
Soukup \cite{MES} showed that the switch Markov chain is rapidly
mixing for half-regular bipartite degree sequences (in which one
part of the bipartition has regular degrees).  An explicit polynomial
bound is not clearly stated.
\item
Cooper, Dyer and Greenhill \cite{CDG} considered regular graphs
and showed that the switch chain is rapidly mixing on the set of all
$d$-regular graphs, for any $d=d(n)$.
Their analysis was extended by Greenhill and Sfragara~\cite{GS} who
adapted the proof to sparse, irregular degree sequences.
Neither of these works explicitly treated bipartite graphs,
though the arguments in both papers are simpler when restricted to
bipartite graphs.
In Corollary~\ref{cor:switch_chain_mixing_semiregular}
we state an upper bound on the mixing time of the switch chain
on $\mathcal{B}(\dvec,k)$ which arises from the arguments of~\cite{CDG,GS}
when $3 \leq \dmax \leq \frac{1}{3}\sqrt{M}$.
Specifically we show that in this range, the switch chain
has mixing time
\[ \Delta^{10} M^7 \big( \nfrac{1}{2} M\log(M) + \log(\varepsilon^{-1})\big)\]
where $\Delta = \max\{ \dmax,\, k\}$.
\item Jerrum and Sinclair~\cite{JS} defined a notion of P-stability
for degree sequences.  Roughly speaking, a degree sequence $\dvec$ is
P-stable if small perturbations to $\dvec$ only change the number of
realisations (graphs with degree sequence $\dvec$) by a small amount.
The notion of P-stable can be adapted to bipartite graphs.
Amanatidis and Kleer~\cite{AK} defined a possibly stronger notion,
\emph{strong stability}, and showed that the switch chain
is rapidly mixing for any strongly stable degree sequence,
and for any strongly stable bipartite degree sequence.
Erd\H{o}s et al.~\cite{hungarians} proved that the switch chain is
rapidly mixing for any P-stable class of bipartite degree sequences.
\end{itemize}
\item[(VI)]The Curveball chain~\cite{Verhelst}
 is another Markov chain for sampling bipartite
graphs with given degrees, in which multiple switches are performed
simultaneously.  Carstens and Kleer~\cite{CK2018} showed that the
Curveball chain is rapidly mixing whenever the switch chain is
rapidly mixing.
\end{itemize}

\bigskip

\begin{remark}
We have focussed on uniform hypergraphs, but our approach can be
adapted to non-uniform hypergraphs.  Given a vector $\kvec=(k_1,\ldots, k_m)$
which stores the desired edge sizes, let $m_{\ell}$ be the number
of edges of size $\ell$, that is,
\[ m_\ell = |\{ j\in [m] \, : \, k_j = \ell\}|\]
for $\ell\geq 2$.
Each hypergraph $H$
on $[n]$ with edge sizes given by $\kvec$ and with degree sequence $\dvec$
corresponds to exactly $\prod_{\ell=2}^n m_\ell!$ bipartite
graphs from $\mathcal{B}(\dvec,\kvec)$, as now we must restrict
to edge labellings $e_1,\ldots, e_m$ so that $|e_j|=k_j$ for $j=1,\ldots, m$.
All of the bipartite graph sampling algorithms mentioned in this
section generalise to $\mathcal{B}(\dvec,\kvec)$, with the exception
of the result by Miklos, Erd\H{o}s and Soukup regarding the switch chain
for half-regular bipartite graphs~\cite{MES}.
\end{remark}

\subsection{Sampling hypergraphs using the configuration model}\label{s:sampling-hypergraphs}

To the best of our knowledge, the only rigorously-analysed algorithm for sampling
hypergraphs with given degrees is the \emph{configuration model}.
The analogue of the configuration model for
hypergraphs has been used by various authors, for example, in the
study of random hypergraphs~\cite{CFMR} and for asymptotic
enumeration~\cite{DFRS}. In the configuration model corresponding
to $\mathcal{H}_k(\dvec)$ there are
$n$ objects, called cells, and the $i$th cell contains $d_i$ (labelled) points.
A \emph{configuration} is a partition of the
$M = \sum_{i=1}^n d_i$ points into $M/k$ parts, each
containing $k$ points.  A random configuration can be chosen in $O(M)$ time.
Shrinking each cell to a node gives an
$k$-uniform hypergraph which may contain loops (that is, an edge
containing a node more than once) or repeated edges.   If the resulting hypergraph
is not simple then the configuration is rejected and a new random configuration
is sampled.
We say that a configuration is simple if the corresponding hypergraph is simple.

Hence, the configuration model can be used for efficient sampling when
the probability that a randomly chosen configuration is simple
is bounded below by the inverse of some polynomial in $n$.
This implies that the expected number of trials before a simple configuration
is found is at most polynomial.

It follows from asymptotic results of
Dudek, Frieze, Ruci{\' n}ski and {\v S}ileikis~\cite{DFRS}
that when $\dvec = (d,d,\ldots, d)$ is regular and $k\geq 3$ is constant,
the expected number of trials before a simple configuration is sampled is
 \[ \exp\,\paren*{\nfrac{1}{2}(k-1)(d-1) + o(1)} \]
assuming that $k=3$ and $d = d(n)=o(n^{1/2})$, or $k\geq 4$ and $d = d(n)=o(n)$.
(Asymptotics are as $n\to\infty$, restricted to values of $n$ such that $dn$ is divisible by
$k$.)
For irregular degrees, let $M_2 = M_2(\dvec) = \sum_{j=1}^n d_j(d_j-1)$.
It follows from~\cite[Corollary~2.3]{BG2} that
the expected number of trials before a simple configuration is sampled is
\[
  \exp\,\paren*{ \frac{(k-1)M_2}{2M} + o(1)}
 \leq \exp\,\paren*{ \nfrac{1}{2}(k-1)(\dmax -1) + o(1) }
\]
whenever
$k^4\dmax^3 = o(M)$.  Here $k=k(n)$ and $\dvec = \dvec(n)$ are such that
$k$ divides $M$ for infinitely many values of $n$.
We collect these facts together into the following lemma.

\begin{lemma}
\label{lem:hypergraph-config}
The configuration model gives an efficient algorithm for sampling
uniformly from $\mathcal{H}_k(\dvec)$ whenever $kd_{\max} = O(\log n)$.
If $kd_{\max} \leq C\log n$ for some constant $C>0$ then
the expected runtime of this algorithm for $\mathcal{H}_k(\dvec)$
is $O(n^C\, M) = O(\dmax\, n^{C+1})$.
If $kd = o(\log n)$ then the expected runtime of this algorithm is $O(M) = O(\dmax n)$.
Note $d_{TV}(\sigma_{\mathcal{H}},\pi_{\mathcal{H}})=0$ as
the output is exactly uniform.
\end{lemma}

Gao and Wormald~\cite{GW} built on earlier work of McKay and
Wormald~\cite{MW90} to give a fast algorithm for exactly uniform sampling
of $d$-regular graphs.  Using a recent improvement of Arman, Gao and
Wormald~\cite{AGW}, a uniformly random $d$-regular graph on $n$
vertices can be generated in expected runtime $O(dn + d^4)$
whenever $d=o(n^{1/2})$,
and a random graph with degree sequence $\dvec$ can be generated in
runtime $O(M)$ whenever $\dmax^4 = O(M)$.  It is likely that this approach
could be adapted to the problem of sampling hypergraphs uniformly.

\section{Analysis of \textsc{HypergraphSampling}}
\label{sec:rejection_sampling}

First we prove Theorem~\ref{thm:algorithm}.

\begin{proof}[Proof of Theorem~\ref{thm:algorithm}]
To prove (i), observe that by definition,
  \begin{align*}
    \sigma_{\mathcal{H}}(H)
  = \sum_{B \in \varphi^{-1}(H)} \mathbb{P}_{\sigma_{\mathcal{B}}}
        \paren*{B \mid B \in \mathcal{B}^*(\boldsymbol{d}, k)}
= \frac{
          1
        }{
          \mathbb{P}_{\sigma_{\mathcal{B}}} \paren*{\mathcal{B}^*(\boldsymbol{d},k)}
        }
        \cdot \sum_{B \in \varphi^{-1}(H)} \sigma_{\mathcal{B}}(B).
  \end{align*}
This equality also holds with $\sigma_{\mathcal{H}}$,
$\sigma_{\mathcal{B}}$ replaced by $\pi_{\mathcal{H}}$,
$\pi_{\mathcal{B}}$, respectively. Since the set of all preimages
$\{ \varphi^{-1}(H)\ :\  H\in\mathcal{H}_k(\dvec)\}$ forms a
partition of $\mathcal{B}^*(\dvec,k)$, and using the triangle
inequality, we have
  \begin{align*}
    d_{TV}\left( \sigma_{\mathcal{H}}, \pi_{\mathcal{H}} \right)
      & = \dfrac{1}{2} \sum_{H \in \mathcal{H}_k(\boldsymbol{d})}
        \abs*{ \,\sigma_{\mathcal{H}}(H) - \pi_{\mathcal{H}}(H) \, } \\
      & \leq
        \dfrac{1}{2}
        \sum_{H \in \mathcal{H}_k(\boldsymbol{d})} \,
        \sum_{B \in \varphi^{-1}(H)} \,\, \abs*{
          \frac{
            \sigma_{\mathcal{B}}(B)
          }{
            \mathbb{P}_{\sigma_{\mathcal{B}}} \paren*{\mathcal{B}^*(\boldsymbol{d},k)}
          }
          - \frac{
            \pi_{\mathcal{B}}(B)
          }{
            \mathbb{P}_{\pi_{\mathcal{B}}} \paren*{\mathcal{B}^*(\boldsymbol{d},k)}
          }
        } \\
      & \leq
        \dfrac{1}{2}
        \sum_{B \in \mathcal{B}^*(\boldsymbol{d},k)} \,\, \abs*{
          \frac{
            \sigma_{\mathcal{B}}(B)
          }{
            \mathbb{P}_{\pi_{\mathcal{B}}} \paren*{\mathcal{B}^*(\boldsymbol{d},k)}
          }
          - \frac{
            \pi_{\mathcal{B}}(B)
          }{
            \mathbb{P}_{\pi_{\mathcal{B}}} \paren*{\mathcal{B}^*(\boldsymbol{d},k)}
          }
        } \\
        &  \ \ \ \ \ \ \ \ \ \ \ \ \ \  + \dfrac{1}{2}
        \sum_{B \in \mathcal{B}^*(\boldsymbol{d},k)} \,\, \abs*{
          \frac{
            \sigma_{\mathcal{B}}(B)
          }{
            \mathbb{P}_{\sigma_{\mathcal{B}}} \paren*{\mathcal{B}^*(\boldsymbol{d},k)}
          }
          - \frac{
            \sigma_{\mathcal{B}}(B)
          }{
            \mathbb{P}_{\pi_{\mathcal{B}}} \paren*{\mathcal{B}^*(\boldsymbol{d},k)}
          }
        } \\
      &  \leq
        \frac{
          d_{TV}\left(\sigma_{\mathcal{B}},\pi_{\mathcal{B}}\right)
        }{
          \mathbb{P}_{\pi_{\mathcal{B}}} \paren*{
            \mathcal{B}^*(\boldsymbol{d},k)
          }
        }
        + \dfrac{1}{2}
        \cdot \abs*{
            \frac{
              1
            }{
              \mathbb{P}_{\sigma_{\mathcal{B}}} \paren*{\mathcal{B}^*(\boldsymbol{d},k)}
            }
            - \frac{
              1
            }{
              \mathbb{P}_{\pi_{\mathcal{B}}} \paren*{\mathcal{B}^*(\boldsymbol{d},k)}
            }
          }
        \sum_{B \in \mathcal{B}^*(\boldsymbol{d},k)}\,\,
          \sigma_{\mathcal{B}}(B)
        \\
      & \leq \dfrac{3}{2}
        \cdot \frac{d_{TV}\left(\sigma_{\mathcal{B}},\pi_{\mathcal{B}}\right)}{\mathbb{P}_{\pi_{\mathcal{B}}} \paren*{\mathcal{B}^*(\boldsymbol{d},k)}}.
\end{align*}
The final inequality follows from applying (\ref{dTV}) with $S=\mathcal{B}^*(\boldsymbol{d},k)$.

Next, (ii) is immediate as the number of times that
\textsc{HyperSampling} calls $\mathcal{A}$ has a geometric
distribution with mean $q(\dvec, k)$, and we test whether $B$ is H-simple once
for every call of $\mathcal{A}$. 
We assume here that the test for H-simplicity also creates the hypergraph 
$\varphi(B)$, so there is no additional cost for the final (output) step of
the algorithm. 
Then, (\ref{dTV}) implies
that
\[ q(\dvec,k)^{-1} = \mathbb{P}_{\sigma_{\mathcal{B}}}\paren*{\mathcal{B}^*(\boldsymbol{d},k)}
                   \geq \mathbb{P}_{\pi_{\mathcal{B}}}\paren*{\mathcal{B}^*(\boldsymbol{d},k)} - \varepsilon,
\]
and (iii) follows.

Finally, suppose that
$\mathcal{A}$ is an FPAUS for $\mathcal{B}(\dvec,k)$ and (\ref{eq:c0})
holds for some $c_0\in (0,1-\varepsilon)$.
It follows from~(ii) and~(iii) that the probability that
\textsc{HypergraphSampling}($\dvec, k,\mathcal{A})$ performs
more than $\lceil 2(1-c_0-\varepsilon)^{-1}\rceil$ iterations of
$\mathcal{A}$ is at most $e^{-2}\leq \nfrac{1}{4}$.
Therefore, terminating
\textsc{HypergraphSampling}($\dvec, k,\mathcal{A})$ after
this many calls to $\mathcal{A}$ gives an FPAUS for $\mathcal{H}_k(\dvec)$.
To achieve a total variation of $\varepsilon$ from
\textsc{HypergraphSampling}($\dvec, k,\mathcal{A})$, the algorithm
$\mathcal{A}$ should be given input $\varepsilon' = 2\varepsilon\, (1-c_0)/3$,
by~(i).
\end{proof}

\bigskip

A general approach for bounding
$\mathbb{P}_{\pi_{\mathcal{B}}} \paren*{ \mathcal{B}^*(\dvec,k)}$
is given by the following lemma.
The constant $c_1$ in Lemma \ref{lem:simple_graph_probability}
captures the maximum edge probability relative to the uniform case
(in which every neighborhood is equally likely).
If $c_1$ is large then some neighbourhoods are much more likely under $\sigma_{\mathcal{B}}$
than they would be under the uniform distribution.
The extent to which the events ``$\mathcal{N}_B(y) = \mathcal{W}$'' are
negatively correlated, as $y$ varies over $Y$ for fixed $\mathcal{W}\subseteq X$,
is described by the constant $c_2$.
Intuitively, if the degree sequence is close to regular then
we expect both $c_1$ and $c_2$ to be close to one.

\begin{lemma}
  \label{lem:simple_graph_probability}
  Suppose that $k \geq 3$ is an integer, $\boldsymbol{d}$ is a sequence of nonnegative integers such that $\mathcal{B}(\boldsymbol{d}, k)$ is non-empty, and that $B = (X \cup Y, A) \in \mathcal{B}(\boldsymbol{d}, k)$ is a random bipartite graph according to the uniform distribution $\pi_{\mathcal{B}}$. Then, suppose that there are constants $c_1$ and $c_2$ such that for any $y, y' \in Y$ and any subset $\mathcal{W} \subseteq X$ of size $k$,
  \[
    \mathbb{P}_{\pi_{\mathcal{B}}} \paren*{\mathcal{N}_B(y) = \mathcal{W}}
    \leq c_1 \cdot \binom{n}{k}^{-1}
  \]
and
  \[
    \mathbb{P}_{\pi_{\mathcal{B}}} \paren*{
      \mathcal{N}_B(y')= \mathcal{W} \ | \
      \mathcal{N}_B(y) = \mathcal{W}
    }
    \leq c_2
      \cdot \mathbb{P}_{\pi_{\mathcal{B}}} \paren*{\mathcal{N}_B(y') = \mathcal{W}}.
  \]
  Then
  \[
    \mathbb{P}_{\pi_{\mathcal{B}}} \paren*{
      \mathcal{B}^*(\boldsymbol{d}, k)
    }
    \geq 1 - c_1 c_2 \,\binom{m}{2} \, \binom{n}{k}^{-1}.
  \]
\end{lemma}

\begin{proof}
Let $B$ be an element from the set $\mathcal{B}(\boldsymbol{d},k)$
drawn uniformly at random, and let $X^{\{k\}}$ be the set of all
$k$-subsets of $X$. We omit the subscript $\pi_{\mathcal{B}}$ on all
following probabilities. For $1 \leq s < t \leq m$, we define the
random variable
  \[
    Z_{st}
    = \left\{
      \begin{array}{ll} 1 &
        \text{ if } \mathcal{N}_B(y_s) = \mathcal{N}_B(y_t) \\
      0 & \text{ otherwise,}
    \end{array}\right.
  \]
for $y_s, y_t \in Y$, and we let
  \[
    Z = \sum_{1 \leq s < t \leq m} Z_{st}
  \]
be the random variable denoting the number of pairs of nodes
$(y_s,y_t)$ that have the same neighborhood in $B$. Note that
  \begin{equation}
    \label{eqn:irregular_simple_relate_x}
    \mathbb{P} \left(B \in \mathcal{B}^*(\boldsymbol{d}, k) \right)
    = 1 - \mathbb{P}(Z \geq 1).
  \end{equation}
Using the union bound over all possible pairs of nodes
$(y_s,y_t)$,
  \begin{equation*}
    \mathbb{P}(Z \geq 1)
      \leq \sum_{1 \leq s < t \leq m} \mathbb{P}(Z_{st} = 1).
  \end{equation*}
Likewise, for a fixed pair of nodes $(y_s,y_t)$, applying the
union bound over ${X}^{\{k\}}$ shows us that
  \begin{equation*}
    \mathbb{P}(Z_{st} = 1)
      \leq \sum_{\mathcal{W} \in {X}^{\{k\}}}
        \mathbb{P}(\mathcal{N}_B(y_s) = \mathcal{N}_B(y_t) = \mathcal{W}).
  \end{equation*}
Now, using the law of total probability,
  \begin{align}
    \mathbb{P}(Z \geq 1)
      & \leq \sum_{1 \leq s < t \leq m} \,\,\sum_{\mathcal{W} \in {X}^{\{k\}}}
        \mathbb{P}(\mathcal{N}_B(y_s) = \mathcal{N}_B(y_t) = \mathcal{W})
        \nonumber \\
      & = \sum_{1 \leq s < t \leq m} \,\,\sum_{\mathcal{W} \in {X}^{\{k\}}}
        \mathbb{P}(\mathcal{N}_B(y_s) = \mathcal{W} \ | \ \mathcal{N}_B(y_t) = \mathcal{W})
          \cdot \mathbb{P}(\mathcal{N}_B(y_t) = \mathcal{W})
        \nonumber \\
      & \leq c_1 c_2 \,\binom{m}{2} \, \binom{n}{k}^{-1}.
        \label{eqn:irregular_simple_px}
  \end{align}
The proof is completed by combining
\eqref{eqn:irregular_simple_relate_x} and
\eqref{eqn:irregular_simple_px}.
\end{proof}\medskip

\section{Probability of H-simplicity}\label{sec:H-simple}

The expected running time of \textsc{HypergraphSampling} is governed
by the runtime of algorithm $\mathcal{A}$ and the probability that a
randomly chosen element of $\mathcal{B}(\boldsymbol{d},k)$ is
H-simple. In Section \ref{sec:irregular} we provide an asymptotic
estimate which holds when $\boldsymbol{d}$ is irregular and sparse.
In Section \ref{sec:regular} we give a combinatorial argument  for
the case of $d$-regular $k$-uniform hypergraphs. In particular, these
sections yield Theorem \ref{thm:main_irregular} and Theorem \ref{thm:main_regular}, respectively.

\subsection{Irregular, sparse degrees}\label{sec:irregular}

In this section we prove a lower bound on the probability that a
uniformly random graph from $\mathcal{B}(\dvec,k)$ is H-simple,
using an asymptotic formula for irregular, sparse bipartite graphs.
Given a bipartite degree sequence $(\boldsymbol{d},\boldsymbol{k})$, for any positive 
integer $r$ define $M_r = \sum_{i=1}^n (d_i)_r$ and $L_r =
\sum_{j=1}^m (k_j)_r$, where $(a)_b = a(a-1)\cdots (a-b+1)$ denotes
the falling factorial. Let $L = L_1$ and $M = M_1$, and note that $L
= M$ for any graphical bipartite degree sequence. We also let
$\dmax = \max_i d_i$ and $\kmax = \max_j k_j$.

The following asymptotic enumeration result is a simpler, but weaker
restatement of the main theorem from~\cite{GMW}, which is slightly
stronger than that of McKay~\cite{McKay84}.
(It follows from Theorem~\ref{thm:zero_one_matrix_count_precise} that
$\mathcal{B}(\dvec,\kvec)\neq \emptyset$ when $\kmax\dmax = o(M^{2/3})$.)

\begin{theorem} \emph{(\cite[Theorem 1.3]{GMW})}\,
\label{thm:zero_one_matrix_count_precise} Suppose that $M\to\infty$,
and that $\dvec = (d_1,\ldots, d_n)$,\, $\kvec =
(k_1,\ldots, k_m)$ are sequences of nonnegative integer functions of
$M$ which both sum to $M$.
If $\kmax \dmax = o(M^{2/3})$ then
\[
    \abs*{\mathcal{B}(\boldsymbol{d},\kvec)}
    = \frac{M!}{\prod_{i=1}^n d_i!\, \prod_{j=1}^m k_j!} \, \exp\,\Bigg(
      -\frac{M_2 L_2}{2M^2} +O\left(\frac{\dmax^2\kmax^2}{M}\right)\Bigg).
\]
\end{theorem}

Using this enumeration result, we can prove the main result of this
section.   First some useful identities which will be
used in the proof without further comment: if $|\eta| < 1$ then
$\exp(\eta) = 1 + O(\eta)$ and $(1 + \eta)^{-1} = 1 + O(\eta)$.
Also observe that if $k^2=o(M)$ then
\begin{equation}
\label{useful}
\frac{M^k}{(M)_k}\leq \frac{M^k}{(M-k)^k} = \big(1 - k/M\big)^{-k} = \exp\big(O(k^2/M)\big) = 1 + O(k^2/M).
\end{equation}

\begin{theorem}
\label{thm:irregular_edge_correlation} Assume that for each positive
integer $n$ we have an integer $k=k(n)\geq 3$ and a sequence
$\boldsymbol{d} =\boldsymbol{d}(n) = (d_1, \ldots, d_n)$ of positive
integers such that $M = \sum_{i=1}^n d_i$ tends to infinity
with $n$.
Assume that
$k^2\dmax^2 = o(M)$ 
and let $\pi_{\mathcal{B}}$ be the uniform distribution on
$\mathcal{B}(\boldsymbol{d},k)$. Then
  \[
    \mathbb{P}_{\pi_{\mathcal{B}}} \paren*{
      \mathcal{B}^*(\boldsymbol{d}, k)
    }
    \geq 1 -
      \frac{n^k \, \dmax^k}{M^k}
      \cdot \binom{M/k}{2} \, \binom{n}{k}^{-1}
      \cdot (1 + o(1)).
  \]
\end{theorem}

\begin{proof}
Let $m = M / k$, which by assumption is an integer. Then, suppose
that we have some neighbourhood $\mathcal{W} \in X^{\set{k}}$, and
two integers $s, t \in [m]$. We will prove the desired result by
conditioning on the neighbourhoods of $y_s$ and $y_t$ being equal
to $\mathcal{W}$, at which point we can apply Lemma
\ref{lem:simple_graph_probability}. To do this, we first define two
bipartite degree sequences
$(\boldsymbol{d}^{\hspace*{1pt}\prime}, \boldsymbol{k}')$ and
$(\boldsymbol{d}^{\hspace*{1pt}\prime\prime}, \boldsymbol{k}'')$, as follows:
  \[
    d_i^{\hspace*{1pt}\prime} = \left\{
      \begin{array}{ll}
        d_i - 1 & \textrm{ if } x_i \in \mathcal{W}, \\
        d_i & \textrm{ if } x_i \in X \setminus \mathcal{W},
      \end{array}
    \right.
    \quad \text{ and } \ \
    k_{j}' = \left\{
      \begin{array}{ll}
        0 & \textrm{ if } j = s, \\
        k & \textrm{ if } j \in [m] \setminus \set{s},
      \end{array}
    \right.
  \]
and
  \[
    d^{\hspace*{1pt}\prime\prime}_i = \left\{
      \begin{array}{ll}
        d_i - 2 & \textrm{ if } x_i \in \mathcal{W}, \\
        d_i & \textrm{ if } x_i \in X \setminus \mathcal{W},
      \end{array}
    \right.
    \quad \text{ and } \ \
    k_{j}'' = \left\{
      \begin{array}{ll}
        0 & \textrm{ if } j \in \set{s, t}, \\
        k & \textrm{ if } j \in [m] \setminus \set{s,t}.
      \end{array}
    \right.
  \]
We also extend the notation for $M$ and $M_2$ to $\boldsymbol{d}^{\hspace*{1pt}\prime}$
and $\boldsymbol{d}^{\hspace*{1pt}\prime\prime}$ by appending one or two dashes, and likewise
for $L$ and $L_2$. By assumption, $k^2 \dmax^2 = o(M)$, which
implies that $k\dmax = o(M^{1/2})$.
Hence, the conditions of Theorem~\ref{thm:zero_one_matrix_count_precise}
are satisfied, and we can approximate both $\abs*{\mathcal{B}(\boldsymbol{d},\kvec)}$ and
$\abs*{\mathcal{B}(\boldsymbol{d}^{\hspace*{1pt}\prime},\kvec')}$. Considering the ratio of
these, since $d^{\hspace*{1pt}\prime}_i = d_i$ whenever $x_i \notin \mathcal{W}$, and
$k_{j}' = k_{j}$ whenever ${j} \neq s$, many factors cancel, leading
to
  \begin{equation}
    \label{eqn:irregular_edge_correlation_n_unsimplified}
    \frac{
      \abs*{\mathcal{B}(\boldsymbol{d}^{\hspace*{1pt}\prime},\kvec')}
    }{
      \abs*{\mathcal{B}(\boldsymbol{d},k)}
    }
    = \frac{k!}{(M)_k}
      \cdot \prod_{x_i \in \mathcal{W}} d_i
      \cdot \exp\,\paren*{
        \frac{M_2 L_2}{2M^2}
        - \frac{M_2' L_2'}{2(M')^2}
        + O(k^2\dmax^2/M)
      }.
  \end{equation}
Next, by definition of $M_2'$, we see that
\[
    M_2'
       = M_2
        - \sum_{x_i \in \mathcal{W}} (d_i)_2
        + \sum_{x_i \in \mathcal{W}} (d_i - 1)_2
       = M_2 - 2 \sum_{x_i \in \mathcal{W}} (d_i - 1)
       = M_2 \left( 1 - O\left( k \dmax M_2^{-1} \right) \right).
\]
Similarly,
\[
    M' = M-k = M \left( 1 - O \left(k M^{-1} \right) \right),
    \qquad
    L_2' = (k-1)\, M' = L_2 \left(1 - O\left( k M^{-1} \right) \right).
\]
Then
  \begin{align*}
    \frac{M_2 L_2}{2M^2} - \frac{M_2' L_2'}{2(M')^2}
      & \leq \frac{M_2 L_2}{2M^2}
        - \frac{M_2 L_2}{2M^2}
        \cdot \paren*{
           1 - O\paren*{k \dmax M_2^{-1} + kM^{-1} }
        } \\
      & = O\paren*{ k \dmax L_2 M^{-2} + k M_2 L_2 M^{-3} } \\
      & = O\paren*{ \frac{k^2 \dmax}{M} }.
  \end{align*}
The final equality follows as $M_2 \leq \dmax M$ and $L_2 = (k-1)M$.
Therefore, combining the above identities with (\ref{useful}) and
\eqref{eqn:irregular_edge_correlation_n_unsimplified} implies that
  \begin{equation}
    \label{eqn:irregular_edge_cor_ratio_basic}
    \frac{
      \abs*{\mathcal{B} (\boldsymbol{d}^{\hspace*{1pt}\prime},\boldsymbol{k}') }
    }{
      \abs*{\mathcal{B} (\boldsymbol{d},\boldsymbol{k}) }
    }
    = \frac{k!}{M^k}
      \cdot \prod_{x_i \in \mathcal{W}} d_i
      \cdot \big(1 +  O(k^2\dmax^2/M)\big).
  \end{equation}

Next, observe that there is a bijective relationship between
bipartite graphs $B_0 \in\mathcal{B}(\boldsymbol{d}^{\hspace*{1pt}\prime}, \kvec')$,
and bipartite graphs $B \in \mathcal{B}(\boldsymbol{d}, \kvec)$
such that $\mathcal{N}_B(y_s) = \mathcal{W}$, using the map
$B\mapsto B_0 = B \setminus \set{y_s}$ which deletes vertex $y_s$
and reduces the degrees of each neighbour of $y_s$ by~1.
Hence,
\begin{equation}
  \label{eq:new}
    \mathbb{P}_{\pi_{\mathcal{B}}} \paren*{
      \mathcal{N}_B(y_s) = \mathcal{W}
    }
    = \frac{
        \abs*{\mathcal{B}(\boldsymbol{d}^{\hspace*{1pt}\prime},\kvec')}
      }{
        \abs*{\mathcal{B}(\boldsymbol{d},\kvec)}
      }.
\end{equation}
By assumption, $k^2 \dmax^2 = o(M)$, so using
\eqref{eqn:irregular_edge_cor_ratio_basic} we find that
  \begin{equation}
\label{eq:c1}
    \binom{n}{k}
    \cdot \mathbb{P}_{\pi_{\mathcal{B}}} \paren*{\mathcal{N}_B(y_s) = \mathcal{W}}
      \leq \frac{n^k d_{\max}^k}{M^k} \cdot \left(1 + o(1)\right).
  \end{equation}
We can use a similar argument to prove a bound on the conditional edge probability needed for Lemma~\ref{lem:simple_graph_probability}. 
Observe that if $\boldsymbol{d}$ and $\kvec$ satisfy the
conditions of Theorem~\ref{thm:zero_one_matrix_count_precise}
then so do $\boldsymbol{d}^{\hspace*{1pt}\prime\prime}$ and $\kvec''$.
Hence, we can repeat the argument that produces (\ref{eq:new}) followed by the substitutions used above to see that 
  \begin{align*}
    \mathbb{P}_{\pi_{\mathcal{B}}} \paren*{
      \mathcal{N}_B(y_t) = \mathcal{W} \mid
      \mathcal{N}_B(y_s) = \mathcal{W}
    }
    &= \frac{
        \abs*{\mathcal{B}(\boldsymbol{d}^{\hspace*{1pt}\prime\prime},\kvec'')}
      }{
        \abs*{\mathcal{B}(\boldsymbol{d}^{\hspace*{1pt}\prime},\kvec')}
      }\\
 &= \frac{k!}{(M-k)^k}\cdot \prod_{x_i\in \mathcal{W}} (d_i-1) \cdot
      \big(1 +  O(k^2\dmax^2/M)\big).
  \end{align*}
We will divide this expression by the one given in (\ref{eqn:irregular_edge_cor_ratio_basic}) to obtain, using (\ref{eq:new}),
  \begin{align*}
    \frac{\mathbb{P}_{\pi_{\mathcal{B}}} \paren*{
      \mathcal{N}_B(y_t) = \mathcal{W} \mid
      \mathcal{N}_B(y_s) = \mathcal{W}
    }
    }{\mathbb{P}_{\pi_{\mathcal{B}}} \paren*{
      \mathcal{N}_B(y_t) = \mathcal{W}
    }}
    & = \frac{(M)_k}{(M-k)_k} \cdot
      \prod_{x_i \in \mathcal{W}} \paren*{1 - \frac{1}{d_i}}
      \cdot \big(1 + O(k^2d_{\max}^2/M)\big)\\
    & \leq \frac{(M)_k}{(M-k)_k}
      \cdot \big(1 + O(k^2d_{\max}^2/M)\big)\\
    &= 1 + o(1),
  \end{align*}
using (\ref{useful}) and the assumption that $k^2\dmax^2 =
o(M)$.
From the above inequality and (\ref{eq:c1}),
we can apply Lemma~\ref{lem:simple_graph_probability} with
  \[
    c_1 = \frac{n^k d_{\max}^k}{M^k} \cdot \left(1 + o(1)\right)
    \qquad \textrm{and} \qquad
    c_2 =  1 + o(1),
  \]
to complete the proof.
\end{proof}

\subsection{Regular degrees}
\label{sec:regular}

In this section we present a combinatorial argument to establish a
lower bound on the probability that a uniformly random graph from
$\mathcal{B}(\dvec,k)$ is H-simple, when $\dvec=(d,d,\ldots, d)$ is
regular. We first prove a `sensitivity result' for the set of all bipartite graphs with given degrees. We show that adjusting the degrees on one side of the bipartition, to make them closer to regular, can only
increase the number of bipartite graphs.  It is possible that this
result is known, though we could not find it in the literature. We give a proof in Section~\ref{sec:prop} for completeness.

\begin{proposition}
\label{prop:balanced} 
Let $n, m \in \N$ and let
$(\boldsymbol{d},\kvec)$ be a bipartite degree sequence for
the bipartition $X = \set{x_1,\dots,x_n}$ and $Y =
\{y_1,\dots,y_m\}$. Suppose that we have integers $g,h \in [n]$
such that $d_g \geq d_h + 2$ and define $\boldsymbol{d}'$ by
  \[
    d_i' = \left\{
      \begin{array}{ll}
        d_i - 1 & \text{ if } i = g \\
        d_i + 1 & \text{ if } i = h \\
        d_i & \text{ if }  i \in [n] \setminus \{g,h\}.
      \end{array}
    \right.
  \]
Then
  \[
    \abs*{\mathcal{B}(\boldsymbol{d},\kvec)}
      \leq \abs*{ \mathcal{B}(\boldsymbol{d}',\kvec) }.
  \]
\end{proposition}

Using this proposition, we now prove Theorem~\ref{thm:main_regular}.

\begin{proof}[Proof of Theorem \ref{thm:main_regular}]
Let $B = (X \cup Y, A) \in \mathcal{B}(n,d,k)$ so that all nodes in
$X$ are $d$-regular and all nodes in $Y$ are $k$-regular. Throughout
the proof we let $\mathcal{W} \in X^{\set{k}}$ be a fixed
neighbourhood of size $k$, and we consider two fixed nodes $y_s, y_t \in Y$.

We will prove the desired result by conditioning on the
neighbourhoods of $y_s$ and $y_t$ being equal to $\mathcal{W}$,
at which point we can apply Lemma
\ref{lem:simple_graph_probability}. To do this, let $\mathcal{U} \in
X^{\set{k}}$ be any $k$-subset of $X$. We will analyse
  \begin{equation}
    \label{eqn:prob_new_nhood_conditional}
    \mathbb{P}_{\pi_{\mathcal{B}}} \paren*{
      \mathcal{N}_B(y_s) = \mathcal{U} \
      \mid \  \mathcal{N}_B(y_t) = \mathcal{W}
    }.
  \end{equation}
Our goal will be to show that \eqref{eqn:prob_new_nhood_conditional} achieves a minimum at $\mathcal{U} = \mathcal{W}$. 
Let $\vartriangle$ denote the symmetric difference operator.
Given $\mathcal{W}=\mathcal{N}_B(y_t)$, for any subset
$\mathcal{U} \subseteq X^{\set{k}}$ of size $k$, we define a new bipartite
degree sequence $(\boldsymbol{d}^{\mathcal{U}},
\kvec^{\mathcal{U}})$ by
  \[
    d_i^{\mathcal{U}} = \left\{
      \begin{array}{ll}
        d_i - 2 & \textrm{ if } x_i \in \mathcal{U} \cap \mathcal{W} \\
        d_i - 1 & \textrm{ if } x_i \in \mathcal{U} \vartriangle \mathcal{W} \\
        d_i & \textrm{ if } x_i \in X \setminus (\mathcal{U} \cup \mathcal{W})
      \end{array}
    \right.
    \quad \text{ and } \ \
    k_j^{\mathcal{U}} = \left\{
      \begin{array}{ll}
        0 & \textrm{ if } y_j \in \set{y_s,y_t} \\
        k & \textrm{ if } y_j \in Y \setminus \set{y_s,y_t}
      \end{array}
    \right..
  \]
Now, when $\mathcal{U} \in X^{\set{k}}$ and $\abs{\mathcal{U}
\vartriangle \mathcal{W}} > 0$, we can select a node $x_g \in
\mathcal{U} \setminus \mathcal{W}$ and $x_h \in \mathcal{W}
\setminus \mathcal{U}$, and create a new neighbourhood $\mathcal{U}'
= \paren{\mathcal{U} \cup \set{x_h}} \setminus \set{x_g}$. Then, we
see that $d^{\mathcal{U}'}$ is equal to $d^{\mathcal{U}}$ whenever
$i \notin \set{g,h}$, and that $d^{\mathcal{U}}$ is a more locally
balanced degree sequence than $d^{\mathcal{U}'}$. Hence,
$d^{\mathcal{U}'}$ and $d^{\mathcal{U}}$ satisfy the conditions of
Proposition~\ref{prop:balanced}, and applying Proposition~\ref{prop:balanced}
we conclude that
\[
    \abs*{
      \mathcal{B}(\boldsymbol{d}^{\mathcal{U}'}, \kvec^{\mathcal{U}'})
    }
    \leq \abs*{
      \mathcal{B}(\boldsymbol{d}^{\mathcal{U}}, \kvec^{\mathcal{U}})
    }.
\]
Since $\abs{\mathcal{U} \vartriangle \mathcal{W}} > 0$ is at most $2k$,
we can repeat the above process a finite number of times to conclude
that for any $\mathcal{U} \in X^{\set{k}}$,
  \begin{equation}
    \label{eqn:udash_smaller_than_W}
    \abs*{
      \mathcal{B}(\boldsymbol{d}^{\mathcal{W}}, \kvec^{\mathcal{W}})
    }
    \leq \abs*{
      \mathcal{B}(\boldsymbol{d}^{\mathcal{U}}, \kvec^{\mathcal{U}})
    }.
  \end{equation}
Since
  \[
    \frac{
      \abs*{
        \mathcal{B}(\boldsymbol{d}^{\mathcal{U}}, \boldsymbol{k}^{\mathcal{U}})
      }
    }{
      \abs*{
        \mathcal{B}(\boldsymbol{d}, \boldsymbol{k})
      }
    }
    =  \mathbb{P}_{\pi_{\mathcal{B}}} \paren*{
      \mathcal{N}_B(y_t) = \mathcal{W}
    }\cdot \mathbb{P}_{\pi_{\mathcal{B}}} \paren*{
      \mathcal{N}_B(y_s) = \mathcal{U} \
      \mid \  \mathcal{N}_B(y_t) = \mathcal{W},
    }
  \]
the inequality in \eqref{eqn:udash_smaller_than_W} implies that for any $\mathcal{U},
\mathcal{W} \in X^{\set{k}}$, we have
  \begin{equation}
    \label{eqn:probW_minimum}
    \mathbb{P}_{\pi_{\mathcal{B}}} \paren*{
      \mathcal{N}_B(y_s) = \mathcal{W} \
      \mid \  \mathcal{N}_B(y_t) = \mathcal{W}
    }
    \leq \mathbb{P}_{\pi_{\mathcal{B}}} \paren*{
      \mathcal{N}_B(y_s) = \mathcal{U} \
      \mid \  \mathcal{N}_B(y_t) = \mathcal{W}
    }.
  \end{equation}
But $X^{\set{k}}$ has size $\binom{n}{k}$, so summing over all
possible choices for $\mathcal{U} \in X^{\set{k}}$ in
\eqref{eqn:probW_minimum} shows us that (since the probabilities on
the right must sum to unity)
\[
    \mathbb{P}_{\pi_{\mathcal{B}}} \paren*{
      \mathcal{N}_B(y_s) = \mathcal{W} \
      \mid \  \mathcal{N}_B(y_t) = \mathcal{W}
    }
    \leq \binom{n}{k}^{-1}.
\]
Since $\boldsymbol{d}$ and $\kvec$ are both regular, by
symmetry every possible $\mathcal{W} \in X^{\set{k}}$ is equally
likely; hence
  \[
    \mathbb{P}_{\pi_{\mathcal{B}}} \paren*{
      \mathcal{N}_B(y_t) = \mathcal{W}
    }
    = \binom{n}{k}^{-1}.
  \]
Thus we can apply Lemma \ref{lem:simple_graph_probability} with $c_1
= c_2 = 1$ to conclude that
  \begin{equation*}
    \mathbb{P}_{\pi_{\mathcal{B}}} \paren*{
      \mathcal{B}^*(\boldsymbol{d}, k)
    }
    \geq 1 - \binom{m}{2}\binom{n}{k}^{-1}.
  \end{equation*}
This completes the proof.
\end{proof}

\subsubsection{Sensitivity result for bipartite degree sequences}
\label{sec:prop}

In this section we prove Proposition~\ref{prop:balanced}.

\begin{proof}[Proof of Proposition~\ref{prop:balanced}]
Let $n, m \in \mathbb{N}$ and let $(\dvec, \kvec)$ be a bipartite degree sequence for the bipartition $X, Y$. Suppose that $g, h \in [n]$ are chosen such that $d_g \geq d_h + 2$. We emphasise that $g$ and $h$ are fixed throughout the proof. Let $\vartriangle$ denote the symmetric difference operator. For any $B = (V,E)$ such that $B \in \mathcal{B}(\dvec, \kvec)$, let
  \[ 
    R_2(B) = \mathcal{N}_B(x_g) \cap \mathcal{N}_B(x_h),
    \quad
    R_1(B) = \mathcal{N}_B(x_g) \vartriangle \mathcal{N}_B(x_h),
  \]
and let $S(B)$ be defined by
  \[ 
    S(B) 
      = \paren*{\set{x_g, x_h} \times R_2(B)} 
        \:\cup\: \set{e \in E \colon e \cap \set{x_g, x_h} = 0}. 
   \]
We see that $S(B)$ is the set of all edges from $B$ which contain vertices $y \in Y$ that are adjacent to both $x_g$ and $x_h$, or adjacent to neither of these vertices. Next, let $\mathcal{S}$ be the collection of possible values for $S(B)$:
  \[
    \mathcal{S}(\dvec, \kvec) = \set{S(B) \colon B \in \mathcal{B}(\dvec, \kvec)}.
  \]
Observe that $\mathcal{S}(\dvec,\kvec)\subseteq \mathcal{S}(\dvec',\kvec)$,
since in any $B\in\mathcal{B}(\dvec,\kvec)$ with $S(B)=S_0$, we may
create a bipartite graph $B'\in\mathcal{B}(\dvec',\kvec)$ with $S(B')=S_0$
by deleting the edge $x_g\, y$ and replacing it with $x_h\, y$, for some $y\in \mathcal{N}_B(x_g)$.

For any $S_0 \in \mathcal{S}(\dvec, \kvec)$, we let $\mathcal{T}(S_0; \dvec, \kvec)$ be the subset of $\mathcal{B}(\dvec, \kvec)$ defined by 
  \[
    \mathcal{T}(S_0; \dvec, \kvec) 
    = \set{B \in \mathcal{B}(\dvec, \kvec) \colon S(B) = S_0}.
  \]
Then the set $\{ \mathcal{T}(S_0; \dvec, \kvec) \colon S_0 \in \mathcal{S}(\dvec, \kvec) \}$ forms a partition of $\mathcal{B}(\dvec, \kvec)$. Similarly, the set $\{ \mathcal{T}(S_0; \dvec', \kvec) \colon S_0 \in \mathcal{S}(\dvec, \kvec) \}$ 
forms a partition of $\mathcal{B}(\dvec', \kvec)$. 

We claim that $\abs{\mathcal{T}(S_0; \dvec, \kvec)} \leq \abs{\mathcal{T}(S_0; \dvec', \kvec)}$ for every $S_0 \in \mathcal{S}(\dvec, \kvec)$. 
To see this, first note that for any $B \in \mathcal{B}(\dvec, \kvec)$, the sets $R_1(B)$ and $R_2(B)$ depend only on $S(B)$. Hence we will instead write $R_1(S_0)$, $R_2(S_0)$ for these sets, where $S_0 = S(B)$. Further, for every $B \in \mathcal{T}(S_0; \dvec, \kvec)$ the degree of vertex $x_i \in X$ in $E \setminus S_0$ is
  \[
    \begin{cases}
        d_i - \abs{R_2(S_0)} & \textrm{ if } i \in {g,h}, \\
        0 & \textrm{ otherwise, }
      \end{cases}
  \]
and vertex $y_j \in Y$ has degree $1$ if $j \in R_1(S_0)$, or $0$ otherwise. Next, let $S_0'$ be any set of bipartite edges satisfying this degree sequence. Then $S_0 \cap S_0' = \emptyset$ by construction, and the combined bipartite degree sequence of $S_0 \cup S_0'$ is $(\dvec, \kvec)$. Hence, the bipartite graph with edges $S_0 \cup S_0'$ is an element of $\mathcal{T}(S_0; \dvec, \kvec)$, and every element of
$\mathcal{T}(S_0;\dvec,\kvec)$ is created exactly once in this way. 
It follows
that for all $S_0 \in \mathcal{S}(\dvec,\kvec)$, 
  \[
    \abs{\mathcal{T}(S_0; \dvec, \kvec)} 
       = \binom{d_g + d_h - 2\abs{R_2(S_0)}}{d_g - \abs{R_2(S_0)}} 
       \leq \binom{d_g' + d_h' - 2\abs{R_2(S_0)}}{d_g' - \abs{R_2(S_0)}} 
       = \abs{\mathcal{T}(S_0; \dvec', \kvec)},
      \]
since $d_g' = d_g - 1 \geq  d_h + 1 = d_h'$. This verifies the claimed inequality. Finally, we see that
  \[
    \abs{\mathcal{B}(\dvec,\kvec)}
      = \sum_{S_0 \in \mathcal{S}(\dvec, \kvec)} \abs{\mathcal{T}(S_0; \dvec, \kvec)}
      \leq \sum_{S_0 \in \mathcal{S}(\dvec, \kvec)} \abs{\mathcal{T}(S_0; \dvec', \kvec)}
      \leq \abs{\mathcal{B}(\dvec',\kvec)},
  \]
which completes the proof of Proposition~\ref{prop:balanced}.
\end{proof}

\subsection*{Acknowledgements}
We are grateful to the referees for their helpful comments.

Research was supported by the Netherlands Organisation for
Scientific Research (NWO) through Gravitation Programme Networks
024.002.003. Martin Dyer is supported by the EPSRC research grant
EP/S016562/1 ``Sampling in hereditary classes''.
Catherine Greenhill is supported by the
Australian Research Council Discovery Project DP190100977.

Pieter Kleer acknowledges that part of this work was carried out while he was a postdoctoral fellow at the Max Planck Institute for Informatics in Saarbr\'ucken, Germany.

\appendix
\section{Mixing time bounds for the switch Markov chain}
\label{app:switch} \color{red}

\color{black}

The switch Markov chain was analysed by Cooper, Dyer and
Greenhill~\cite{CDG,CDG-corrigendum} for regular graphs, and by Greenhill and
Sfragara~\cite{GS} for irregular graphs that are relatively sparse.
Some situations which lead to additional factors in the mixing time
bounds from~\cite{CDG,CDG-corrigendum,GS} do not arise in bipartite graphs, and so
it is possible to improve the mixing time bounds for bipartite
graphs. To the best of our knowledge, these bounds have not been
presented elsewhere, so we write them down here.  The proofs
from~\cite{CDG,CDG-corrigendum,GS} use the multicommodity flow method, and are quite
long and technical. We do not give full details, but rather explain
how the proofs from~\cite{CDG,CDG-corrigendum,GS} can be adapted to the bipartite
setting, and give the resulting mixing time bounds. For all notation
that is not defined, and all other missing details, we refer the reader to~\cite{CDG,CDG-corrigendum,GS}.

We begin with regular bipartite graphs, where there are $n$ nodes in
each side of the bipartition and all nodes have degree $d$. (We
stress that this case is not particularly relevant for the
problem of sampling hypergraphs, unless the hypergraph is
$d$-regular and $d$-uniform.) For regular bipartite graphs, the
bound on the bipartite switch chain is a factor of $\nfrac{1}{32}\, d^6 n^2$
smaller than the general (not-necessarily-bipartite) case. (The
constant factor 32 in Theorem \ref{CDGbipartite} below arises since graphs in
$\mathcal{B}(n,d,d)$ have $2n$ nodes.)

\begin{theorem} \label{CDGbipartite}
Let $\mathcal{B}(n,d,d)$
be the set of $d$-regular bipartite graphs with $2n$ nodes and a
given bipartition.  Then the mixing time of the bipartite switch
chain on $\mathcal{B}(n,d,d)$ satisfies
\[ \tau(\varepsilon) \leq 32 d^{17} n^6 \big(2dn\log(2dn) + \log(\varepsilon^{-1})\big).\]
\end{theorem}

\begin{proof}
Temporarily, write $N = 2n$, for ease of comparison with~\cite{CDG,CDG-corrigendum}.
The flow can be defined in exactly the same way as in~\cite{CDG},
though a shortcut edge will never be needed.    (Every  circuit
decomposes into 1-circuits.) Hence there will be at most 3 defect
edges, two labelled $-1$ and one labelled 2, all incident with the
``start vertex'' $x_0$ of the 1-circuit. At most two switches are needed
to transform any encoding into a graph, and~\cite[Lemma 4]{CDG}
becomes $|\mathcal{L}(Z)|\leq 2d^4 N^3 |\mathcal{B}(n,d,d)|$. (This is
smaller by a factor of $d^2 N^2$ than in the non-bipartite case,
essentially because we save one (-1)-switch in the worst case, which
costs $d^2 N^2$.)

In the bipartite case, 
we save a factor of $d^4$ compared with~\cite[Lemma 1]{CDG-corrigendum} 
(which is a correction of~\cite[Lemma 5]{CDG}).  This is because there
can be at most 10 bad pairs in the yellow-green colouring, not 14.
Recall that each bad pair contributes a factor of $d$. (There are at most 3
bad pairs from each defect edge, plus at most one additional bad pair from
wrapping around at $x_0$.  Alternatively, we no longer have a shortcut edge,
which in~\cite{CDG,CDG-corrigendum} was responsible for up to 4 bad pairs,
so 14 goes down to 10.)

Combining these effects, we save a factor of $d^6 N^2$ compared to
the mixing time from~\cite[Theorem~1]{CDG-corrigendum}.  Note that we have
$\ell(f)\leq dN/2$ and, as there are $dn$ edges in any element of $\mathcal{B}(n,d,d)$,
\[ 1/Q(e) \leq 4 \binom{nd}{2}\, |\mathcal{B}(n,d,d)|
  \leq \nfrac{1}{2} d^2 N^2 \, |\mathcal{B}(n,d,d)|, \]
which saves an additional factor of $\nfrac{1}{2}$ compared
with~\cite{CDG,CDG-corrigendum}.
\end{proof}

Let $\Delta = \max\{\dmax,\kmax\}$. By adapting the analysis
from~\cite{GS}, we can prove a bound in the irregular
case which is a factor of $\Delta^4 M^2$ smaller than in the general
(not-necessarily-bipartite) case.

\begin{theorem} \label{GS-bipartite_app}
Let $\mathcal{B}(\dvec,\kvec)$ be the set of all bipartite graphs
with a given node bipartition, degrees $\dvec$ on the left and
degrees $\kvec$ on the right.  Suppose that all degrees are at
least~1 and that $3\leq \dmax, \kmax \leq \nfrac{1}{3}
\sqrt{M}$.  Then the mixing time of the bipartite switch chain
satisfies
\[ \tau(\varepsilon) \leq \Delta^{10} M^7 \big( \nfrac{1}{2} M\log(M) +
  \log(\varepsilon^{-1})\big). \]
\end{theorem}

\begin{proof}
We have $1/Q(e) \leq M^2 \, |\mathcal{B}(\dvec,\kvec)|$ and
$\ell(f)\leq M/2$, as in~\cite[Theorem~1.1]{GS}. Arguing as above,
the number of bad pairs is at most 10, not 14, saving a factor of
$\Delta^4$.  The main thing is the critical lemma~\cite[Lemma 2.5]{GS}, where we give an
upper bound on the number of encodings.  We claim that
\[ |\mathcal{L}^*(Z)| \leq 2 M^4 |\mathcal{B}(\dvec,\kvec)|\]
so long as $\Delta \leq \nfrac{1}{3} \sqrt{M}$, say.  

In~\cite[Lemma 2.5]{GS}, we only performed a $(-1,2)$-switching if we had
four defect edges.  This was to ensure that we definitely had a
$(-1)$-defect edge incident with a 2-defect edge: but when there is
no shortcut edge, this is already guaranteed when we have three
defect edges. Therefore, letting
\[ a = 2\Delta^2 M,\qquad b = 2\Delta^2,\qquad c = \nfrac{9}{8} M^2\]
be the upper bounds that we proved in~\cite[Lemma 2.5]{GS} on the
various ratios, we obtain
\[ |\mathcal{L}^*(Z)| \leq (1 + b + c + bc + c^2 + ac)\, |\mathcal{B}(\dvec,\kvec)|.\]
(The saving here is replacing $bc^2$ by $ac$, and in omitting the
terms involving $b^2$ or $abc$ or $ac^2$, which were needed
in~\cite{GS} to deal with the shortcut edge.)  Using the bounds
$3\leq \Delta$ and $\Delta^2 \leq M/9$, we see that
\[ |\mathcal{L}^*(Z)|\leq 2M^4.\]
This is a factor of $M^2$ smaller than the corresponding bound given
in~\cite[Lemma 2.5]{GS}, again because (in the worst case) we save a
$(-1)$-switch, which gives a ratio of $\nfrac{9}{8}M^2$.  Combining
this with the earlier saving of $\Delta^4$, we obtain the stated
bound.
\end{proof}

The following corollary is most relevant to sampling uniform
hypergraphs with given degrees. It follows directly from Theorem \ref{GS-bipartite_app} by considering a regular sequence $\kvec$.

\begin{corollary}
\label{cor:switch_chain_mixing_semiregular}
Let $\mathcal{B}(\dvec,k)$ be the set of all bipartite graphs with a
given node bipartition, degrees $\dvec = (d_1,\ldots, d_n)$ on one
side and with $m$ nodes of degree $k$ on the other. Let $M = km =
\sum_{j=1}^n d_j$. Suppose that all degrees are at least~1 and that
$3\leq \dmax, k\leq \nfrac{1}{3} \sqrt{M}$. Then the mixing time
of the bipartite switch chain satisfies
\[ \tau(\varepsilon) \leq \Delta^{10} M^7 \big( \nfrac{1}{2} M\log(M) +
  \log(\varepsilon^{-1})\big) \]
where $\Delta = \max\{ \dmax,\, k\}$.
\end{corollary}

\end{document}